\documentclass[pra,showpacs,graphics,twocolumn,floatfix,mathbbm,a4paper,nofootinbib]{revtex4-1}
\usepackage{amsmath,amsthm,amsfonts,amssymb,color,graphicx,times,bbm,cancel}

\newtheorem*{theorem*}{Theorem}

\definecolor{myurlcolor}{rgb}{0,0,0.4}
\definecolor{mycitecolor}{rgb}{0,0.5,0}
\definecolor{myrefcolor}{rgb}{0.5,0,0}
\usepackage[pagebackref]{hyperref}
\hypersetup{colorlinks,
linkcolor=myrefcolor,
citecolor=mycitecolor,
urlcolor=myurlcolor}

\makeatletter
\renewcommand\paragraph{%
 \@startsection{paragraph}{4}{1pc}%
 {0pt}%
 {-4pt}%
 {\normalfont\normalsize\bfseries}}
\makeatother

\begin{document}

\bibliographystyle{apsrev}

\newcommand{\Tr}{\mathrm{Tr}}

\title{Quantum realization of arbitrary joint measurability structures}

\author{Ravi Kunjwal}
\email{rkunj@imsc.res.in} 
\affiliation{The Institute of Mathematical
  Sciences, C.I.T Campus, Tharamani, Chennai 600 113, India.}

\author{Chris Heunen}
\email{heunen@cs.ox.ac.uk}
\affiliation{Department of Computer Science, University of Oxford}

\author{Tobias Fritz}
\email{tfritz@perimeterinstitute.ca} 
\affiliation{Perimeter Institute for Theoretical Physics, Waterloo, Ontario, Canada.}

\date{\today}

\begin{abstract}
In many a traditional physics textbook, a quantum measurement is defined as a projective measurement represented by
a Hermitian operator. In quantum information theory, however, the concept of a measurement is dealt with in complete generality
and we are therefore forced to confront the more general notion of positive-operator valued measures (POVMs) which 
suffice to describe all measurements that can be implemented in quantum experiments. We study the (in)compatibility
of such POVMs and show that quantum theory realizes all possible (in)compatibility relations among sets 
of POVMs. This is in contrast to the restricted case of projective measurements for which commutativity is essentially equivalent to 
compatibility. Our result therefore points out a fundamental feature with respect to the (in)compatibility of 
quantum observables that has no analog in the case of projective measurements.
\end{abstract}


\maketitle

\section{Introduction}
In the traditional textbook treatment of measurements in quantum theory one usually comes across projective measurements. 
For these measurements, commutativity of the associated Hermitian operators is necessary and sufficient for them to be compatible.
That is, commuting Hermitian operators represent quantum observables that can be jointly measured in a single experimental setup.
Furthermore, given a set of $N$ projective measurements, commutativity means pairwise commutativity and we have:
pairwise compatibility $\Leftrightarrow$ global compatibility. This equivalence is rather special since it reduces the problem of deciding whether a set of projective measurements is 
compatible to checking that every pair in the set commutes. Operationally, this also means that the measurement statistics 
obtained by performing these measurements sequentially on any preparation of a quantum system is independent of the sequence 
in which the measurements are performed, \textit{e.g.}, if $A$, $B$, $C$ are Hermitian operators that commute pairwise, then the
sequential measurements $ABC$, $ACB$, $BAC$, $BCA$, $CAB$ and $CBA$ are all physically equivalent.

However, once the projective property is relaxed and the resulting positive-operator valued measures (POVMs) are considered,
the implication ``pairwise compatibility $\Rightarrow$ global compatibility'' no longer holds. The converse implication is still 
true. Indeed, one can construct examples where a set of three POVMs is pairwise compatible but there is no global compatibility
between them~\cite{Kraus,LSW,KG,HRF}. With this in mind, our purpose in this paper is to explore whether there really is any 
constraint on the (in)compatibility relations that one could realize between quantum measurements (POVMs). If, for example, 
certain sets of (in)compatibility relations were not allowed in quantum theory then that would point out conceivable 
joint measurability structures that are nevertheless forbidden in nature. A basic understanding of what is allowed and what is
forbidden in a physical theory is essential from a foundational point of view. Indeed, an example that readily comes to mind 
is the impossibility of faster-than-light signalling, a principle that has served as an invaluable guide to ruling out 
theories---and being highly skeptical of putative phenomena---that may suggest the contrary. Likewise, our larger endeavour in this work is to
study the possibilities and limitations of quantum theory with respect to (in)compatibility relations.

\begin{figure}
\includegraphics{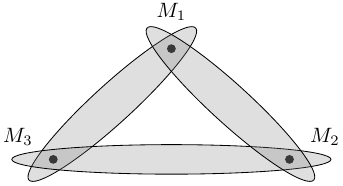}
\caption{Specker's scenario.}
\label{specker}
\end{figure}

It is a fact worth noting that the impossibility of jointly implementing arbitrary sets of measurements is a key ingredient that 
enables a demonstration of the nonclassicality of quantum theory in proofs of Bell's theorem~\cite{Bell64} and the Kochen--Specker 
theorem~\cite{KS67}. A finite set of measurements is called \textit{jointly measurable} or \textit{compatible} if there exists a single
measurement whose various coarse-grainings recover the original measurements. The problem of characterizing the joint measurability 
of observables has been studied in the literature~\cite{heinosaari, SRH}, and at least the joint measurability of binary qubit observables has 
been completely characterized~\cite{BS, YuOh}. The connection between Bell inequality violations and the joint measurability of observables
has also been quantitatively studied~\cite{andersson, wolfetal}.

A natural question that arises when thinking about the (in)compatibility of observables is the following: given a set of 
(in)compatibility relations on a set of vertices representing observables, do they admit a quantum realization? That is, 
can one write down a positive-operator valued measure (POVM) for each vertex such that the (in)compatibility relations 
among the vertices are realized by the assigned POVMs? After formally defining these notions, we answer this question in the 
affirmative by providing an explicit construction of POVMs for any set of (in)compatibility relations. This is our main result.
We will use the terms `(not) jointly measurable' and `(in)compatible' interchangeably in this paper. Part of our motivation in
studying this question comes from the simplest example of joint measurability relations realizable with POVMs but not with 
projective measurements. This joint measurability scenario, referred to as Specker's scenario~\cite{Spe60, LSW, KG}, involves
three binary measurements that can be jointly measured pairwise but not triplewise: that is, for the set of binary measurements
$\{M_1,M_2,M_3\}$, the (in)compatibility relations are given by the collection of compatible subsets $\{\{M_1,M_2\},\{M_2,M_3\},\{M_1,M_3\}\}$.
The remaining nontrivial subset (with at least two observables), namely $\{M_1,M_2,M_3\}$, is incompatible. This can be 
pictured as a hypergraph (Fig.~\ref{specker}).

Specker's scenario has been exploited to violate a generalized noncontextuality inequality using a set of three qubit POVMs 
realizing this scenario~\cite{genNC,LSW,KG}. This novel demonstration of contextuality in quantum theory
raises the question whether there exist other contextuality scenarios---for example in an observable-based hypergraph approach as 
in~\cite{AB,CF}---that do not admit a proof of quantum contextuality using projective measurements, but do admit such a proof using POVMs.
A necessary first step towards answering this question is to figure out what compatibility scenarios are realizable in quantum
theory. One can then ask whether these scenarios allow nontrivial correlations that rule out generalized noncontextuality~\cite{genNC}.
We take this first step by proving that, in principle, all joint measurability hypergraphs are realizable in quantum theory.
The realizability of all joint measurability graphs via projective measurements has been shown recently~\cite{HRF}. This prompted our 
question whether all joint measurability hypergraphs are realizable via POVMs. Our positive answer includes joint measurability
hypergraphs that do not admit a realization using projective measurements. For our construction, it suffices to consider binary observables on
finite-dimensional Hilbert spaces. We start with a more detailed discussion of the relevant concepts.

\section{Definitions}
\paragraph*{POVMs.} A positive-operator valued measure (POVM) on a Hilbert space $\mathcal{H}$ is a mapping $x\mapsto M(x)$ from an outcome set $X$ to the set of positive semidefinite operators
, 
\[
M(x)\in\mathcal{B}(\mathcal{H}),\quad M(x)\geq 0,
\]
such that the POVM elements $M(x)$ sum 
to the identity operator,
\[
\sum_{x\in X}M(x)=I.
\]
If $M(x)^2=M(x)$ for all $x\in X$, then the POVM becomes a ``projection valued measure'', or simply a projective measurement.

\paragraph*{Joint measurability of POVMs.} A finite set of POVMs 
\[
\{M_1,\dots,M_N\},
\]
where measurement $M_i$ has outcome set $X_i$,
is said to be \textit{jointly measurable} or \textit{compatible} if there exists a POVM $M$ with outcome set $X_1\times X_2 \times \dots \times X_N$ that marginalizes to each $M_i$ with outcome set $X_i$, meaning that
\[
M_i(x_i)=\sum_{x_1,\ldots,\cancel{x_i},\ldots,x_N}M(x_1,\dots,x_N) 
\]
for all outcomes $x_i\in X_i$.

\paragraph*{Joint measurability hypergraphs.} A \textit{hypergraph} consists of a set of vertices $V$, and 
a family $E \subseteq \{ e \mid e \subseteq V \}$ of subsets of $V$ called \textit{edges}.
We think of each vertex as representing a POVM, while an edge models joint measurability of the POVMs 
it links. Since every subset of a set of compatible measurements should also be compatible, a joint 
measurability hypergraph should have the property that any subset of an edge is also an edge,
\[
e\in E,\: e'\subseteq e \implies e'\in E.
\]
Additionally, we focus on the case where each edge $e$ is a finite subset of $V$.
This makes a joint measurability hypergraph into an abstract simplicial complex.

Every set of POVMs on $\mathcal{H}$ has such an associated joint measurability hypergraph. 
Hence characterizing joint measurability of quantum observables comes down to figuring out their joint measurability hypergraph.
Our main result solves the converse problem. Namely, ecommutativityvery abstract simplicial complex arises from the joint measurability
relations of a set of quantum observables.

\section{Quantum realization of any joint measurability structure}

\begin{theorem*}
 Every joint measurability hypergraph admits a quantum realization with POVMs.
\end{theorem*}

\begin{proof}
We begin by proving a necessary and sufficient criterion for the joint measurability of $N$ binary POVMs $M_k:=\{E^k_{+},E^k_{-}\}$ of the form
\begin{equation}
\label{Ek}
 E^k_{\pm}:=\frac{1}{2}\left(I\pm\eta\Gamma_{k}\right),
\end{equation}
where the $\Gamma_{k}$ are generators of a Clifford algebra as in the Appendix. The variable $\eta\in[0,1]$ is a purity parameter.
Since $\Gamma_k^2=I$, the eigenvalues of $\Gamma_k$ are $\pm 1$, so that $E^k_\pm$ is indeed positive. The following derivation
of a joint measurability criterion is adapted from a proof first obtained in~\cite{LSW}, and subsequently revised in~\cite{KG},
for the joint measurability of a set of noisy qubit POVMs.
Because $\Gamma_k$ is traceless by~\eqref{traceless}, we can recover the purity parameter $\eta$ as
\[
 \Tr(\Gamma_k E^k_{\pm})=\pm\frac{\eta}{2}d,
\]
so that
\begin{equation}
\label{recovereta}
 \eta=\frac{1}{Nd}\sum_{k=1}^{N}\sum_{x_k\in X_k}\Tr(x_k\Gamma_k E^k_{x_k}),
\end{equation}
where we have introduced one separate outcome $x_k\in X_k:=\{+1,-1\}$ for each measurement $M_k$.

If all 
$M_k=\{E^k_+,E^k_-\}$ together are jointly measurable, then there exists
a joint POVM $M=\{E_{x_1\dots x_N}\}$ satisfying
\[
 E^k_{x_k}=\sum_{x_1,\ldots \cancel{x_k},\ldots x_N}E_{x_1\dots x_N}.
\]
Writing $\vec{x}:=(x_1,\ldots,x_N)$ and $\vec{\Gamma}:=(\Gamma_1,\ldots,\Gamma_N)$, this assumption together with~\eqref{recovereta} implies that
\begin{align*}
 \eta&=\frac{1}{Nd}\sum_{\vec{x}}\Tr\left[\left(\sum_{k=1}^{N}x_k\Gamma_k\right)E_{x_1\dots x_N}\right]\\[5pt]
 &\leq \frac{1}{Nd}\sum_{\vec{x}} \|\vec{x}\cdot\vec{\Gamma}\| \:\Tr\left[E_{\vec{x}}\right]\\[5pt]
 &= \frac{1}{N} \|\vec{x}\cdot\vec{\Gamma}\|,
\end{align*}
where the last step used the normalization $\sum_{\vec{x}} E_{\vec{x}}=I$. Since $(\vec{x}\cdot\vec{\Gamma})^2=\sum_k X_k^2=N\cdot I$ by~\eqref{cliffordprod}, we have $\|\vec{x}\cdot\vec{\Gamma}\|=\sqrt{N}$, and therefore
\[
\eta \leq  \frac{1}{\sqrt{N}} ,
\]
a necessary condition for joint measurability of $M_k$.
To show that this condition is also sufficient, we consider the joint POVM $M=\{E_{\vec{x}}\}$
given by
\begin{equation}
 E_{x_1\dots x_N}:=\frac{1}{2^N} \left(I+\eta\:\vec{x}\cdot\vec{\Gamma}\right).
\end{equation}
We start by showing that this indeed defines a POVM,
\[
E_{x_1\dots x_N}\geq 0, \quad \sum_{x_1,\dots, x_N} E_{x_1\dots x_N}=I.
\]
Positivity follows again from noting that the eigenvalues of $\vec{x}\cdot\vec{\Gamma}$ are $\pm\sqrt{N}$ by~\eqref{cliffordprod}, and normalization from $\sum_{\vec{x}}\vec{x}\cdot\vec{\Gamma}=0$.
Since
\[
 \sum_{x_1,\ldots,\cancel{x_k},\ldots,x_N} E_{x_1\dots x_N}=\frac{1}{2}\left(I+\eta x_k\Gamma_k\right)
\]
coincides with~\eqref{Ek}, we have indeed found a joint POVM marginalizing to the given $M_k$.

Thus $\eta\leq \frac{1}{\sqrt{N}}$ is a necessary and sufficient condition for the joint measurability of $M_1,\ldots,M_N$.

For arbitrary $N$, then, we can construct $N$ POVMs on a Hilbert space of appropriate dimension such that any 
$N-1$ of them are compatible, whereas all $N$ together are incompatible: 
simply take $M_1,\ldots, M_{N}$ from~\eqref{Ek} for any purity parameter $\eta$ satisyfing
\[
 \frac{1}{\sqrt{N}}<\eta\leq\frac{1}{\sqrt{N-1}}.
\]
For example, $\eta=1/\sqrt{N-1}$ will work. The above reasoning guarantees that all $N$ of them together are not compatible, and also that the $M_1,\ldots,M_{N-1}$ are compatible. By permuting the labels and 
observing that the above reasoning did not rely on any specific ordering of the $\Gamma_k$, we conclude that \textit{any} $N-1$ measurements among the $M_1,\ldots,M_N$ are compatible.

What we have established so far is that, 
if we are given any $N$-vertex joint measurability hypergraph 
where every subset of $N-1$ vertices is compatible (\textit{i.e.}\ belongs to a common edge), 
but the $N$-vertex set is incompatible, 
then the above construction provides us with a quantum realization of it.
These ``Specker-like'' hypergraphs are crucial to our construction. For example, for $N=3$, we
obtain a simple realization of Specker's scenario (Fig. \ref{specker}). For $N=2$, we simply obtain a pair of incompatible observables.
Given an arbitrary joint measurability hypergraph, the procedure to construct a quantum realization is now the following:

\begin{enumerate}
 \item Identify the minimal incompatible sets of vertices in the hypergraph. A minimal incompatible set is an incompatible
set of vertices such that any of its proper subsets \textit{is} compatible. In other words, it is a Specker-like hypergraph embedded
 in the given joint measurability hypergraph.
 \item For each minimal incompatible set, construct a quantum realization as above. Vertices that are outside this minimal incompatible set
 can be assigned a trivial POVM in which one outcome is deterministic, represented by the identity operator $I$. Let $\mathcal{H}_i$ denote
 the Hilbert space on which the minimal incompatible set is realized, where $i$ indexes the minimal incompatible sets.
 \item Having thus obtained a quantum representation of each minimal incompatible set, we simply ``stack'' these together in a direct
 sum over the Hilbert spaces on which each of the minimal incompatible sets are realized. On this larger direct sum Hilbert space
 $\mathcal{H}=\oplus_{i} \mathcal{H}_i$, we then have a quantum realization of the joint measurability hypergraph we started with.
\end{enumerate}
For any edge $e\in E$, the associated measurements are compatible on every $\mathcal{H}_i$, and therefore also on $\mathcal{H}$. On the other hand, every $e'\subseteq V$ that is not an edge is contained in some minimal incompatible set (or is itself already minimal), and therefore the associated POVMs are incompatible on some $\mathcal{H}_i$, and 
hence also on $\mathcal{H}$.
\end{proof}

\begin{figure}
\includegraphics[width=0.4\textwidth]{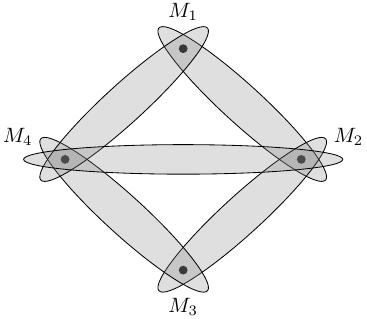}
\caption{A joint measurability hypergraph for $N=4$.}
\label{hyper}
\end{figure}

\begin{figure}
\includegraphics[width=0.5\textwidth]{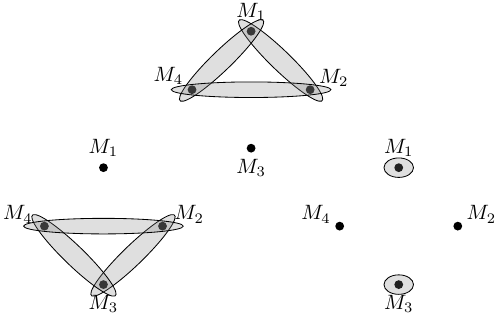}
\caption{Minimal incompatible sets for the joint measurability hypergraph in Fig.~\ref{hyper}.}
\label{incomp}
\end{figure}

\section{A simple example} 
To illustrate these ideas, we construct a POVM realization of a simple joint measurability hypergraph that does not admit a 
representation with projective measurements (Fig.~\ref{hyper}). This hypergraph can be decomposed into three minimal incompatible sets of vertices (Fig.~\ref{incomp}).
Two of these are Specker scenarios for $\{M_1,M_2,M_4\}$ and $\{M_2,M_3,M_4\}$, and the third one is a pair of incompatible vertices
$\{M_1,M_3\}$. For the minimal incompatible set $\{M_1,M_2,M_4\}$, we construct a set of three binary POVMs, $A_k\equiv\{A^k_+,A^k_-\}$ with $k\in \{1,2,4\}$ on a qubit Hilbert space $\mathcal{H}_1$ given by
\begin{equation}
A^k_{\pm}:=\frac{1}{2}\left(I\pm \frac{1}{\sqrt{2}}\Gamma_k\right),
\end{equation}
where the 
matrices $\{\Gamma_1,\Gamma_2,\Gamma_4\}$ can be taken to be the Pauli matrices,
\[
\Gamma_1 = \sigma_z,\quad \Gamma_2=\sigma_x,\quad \Gamma_4=\sigma_y,
\]
similar to~\eqref{pauli}.
The remaining vertex $M_3$ can be taken to be the trivial POVM $A_3=\{0,I\}$ on $\mathcal{H}_1$. A similar construction works for the second 
Specker scenario $\{M_2,M_3,M_4\}$ 
by setting
$B_k:=\{B^k_+,B^k_-\}$ with $k\in \{2,3,4\}$ to be
\begin{equation}
B^k_{\pm}:=\frac{1}{2}\left(I\pm \frac{1}{\sqrt{2}}\Gamma_k\right),
\end{equation}
where
\[
\Gamma_2=\sigma_z,\quad \Gamma_3=\sigma_x,\quad \Gamma_4=\sigma_y
\]
act on another qubit Hilbert space $\mathcal{H}_2$. The remaining vertex $M_1$
can be assigned the trivial POVM, $B_1=\{0,I\}$. The third minimal incompatible set $\{M_1,M_3\}$ can similarly 
be obtained on another qubit Hilbert space $\mathcal{H}_3$ as $C_k:=\{C^k_+,C^k_-\}$, with $k\in \{1,3\}$, given by
\begin{equation}
C^k_{\pm}:=\frac{1}{2}(I\pm \Gamma_k),
\end{equation}
where now \textit{e.g.}~$\Gamma_1=\sigma_z$ and $\Gamma_3=\sigma_x$. 
The remaining vertices $M_2$ and $M_4$ can both be 
assigned the trivial POVM $C_2=C_4:=\{0,I\}$ on $\mathcal{H}_3$.

In the direct sum Hilbert space
$\mathcal{H}:=\mathcal{H}_1\oplus \mathcal{H}_2\oplus \mathcal{H}_3$, we then have a POVM realization 
of the joint measurability hypergraph of Fig.~\ref{hyper}, given by
\[
 M^k_{\pm}:=A^k_{\pm}\oplus B^k_{\pm}\oplus C^k_{\pm}.
\]

\section{Discussion} We have shown, by construction, that any conceivable set of (in)compatibility relations for any number
of quantum measurements can be realized using a set of binary POVMs. Our result thus demonstrates that quantum theory is not 
constrained to admit only a restricted set of (in)compatibility relations, such as those where pairwise compatibility $\Leftrightarrow$
global compatibility, which is the case with projective measurements. Indeed, quantum theory admits all possible (in)compatibility
relations. With respect to (in)compatibility relations, therefore, quantum theory is as far away from classical theories 
(where there are no incompatibilities) as possible. By ``classical theories'' we mean those where all measurements commute.

Although our simple construction works for all joint measurability hypergraphs, it is probably not the most efficient one for
a given joint measurability hypergraph: for Fig.~\ref{hyper}, our representation lives on a six-dimensional Hilbert space. 
For a joint measurability hypergraph with a fixed number of vertices, the dimension of the Hilbert space $\mathcal{H}$ on which our construction
is realized depends on the number of minimal incompatible sets in the hypergraph: that is, $\dim \mathcal{H}=\sum_i \dim \mathcal{H}_i$,
where $\mathcal{H}_i$ is the Hilbert space on which the $i$th minimal incompatible set is realized. 
It remains open what the most efficient construction---requiring the smallest Hilbert space dimension---for a given joint 
measurability hypergraph is. Concerning quantum contextuality, in future work we intend to study whether our sets of POVMs can
lead to nonclassical correlations in the scenarios associated with the underlying joint measurability hypergraphs. This will open
up new avenues for exploiting the nonclassicality of quantum correlations in potential information-theoretic tasks. On the 
theoretical side, our result also opens the door to the use in quantum contextuality of homology theory, matroid theory, and
other powerful combinatorial machinery that relies on hypergraphs, and vice versa. Another potential application of our result
could be in situations where the (in)compatibility of observables is a resource for some task: for example, in such scenarios
one could require a set of measurements to satisfy a specific set of (in)compatibility relations to be useful for the task at
hand and our construction may then offer a way to realize those (in)compatibility relations.

Quite independent of potential applications, our result is of foundational significance for physics since it captures all conceivable (in)compatibility relations within the framework
of quantum theory. We have shown that POVMs allow joint measurability structures that have no analog when thinking of projective measurements alone, and in doing so our contribution
sheds light on the structure of quantum theory and what it really allows us to do.

\section*{Appendix: Clifford algebras}
A \textit{Clifford algebra} consists of a finite set of hermitian matrices $\Gamma_1,\ldots,\Gamma_N$ 
satisfying the relations\footnote{Strictly speaking, this is a \textit{representation} of a Clifford algebra, but the difference 
between algebras and their representations is not relevant here.}

\begin{equation}
\label{acrel}
\Gamma_j\Gamma_k + \Gamma_k\Gamma_j = 2\delta_{jk}I,
\end{equation}
Clifford algebras are the mathematical structure behind the definition of spinors and the Dirac equation~\cite{Lounesto}. 
They can be constructed recursively as follows~\cite[Sec.~16.3]{Lounesto}. Given $\Gamma_1,\dots,\Gamma_N$ living on a Hilbert
space $\mathcal{H}_{N}$, one obtains $\Gamma_1,\ldots,\Gamma_{N+2}$ on $\mathcal{H}_N\otimes \mathbb{C}^2$ by the following rules.
\begin{enumerate}
\item For each $i=1,\dots,N$, substitute
\[
\Gamma_i \rightarrow \Gamma_i\otimes\sigma_z.
\]
\item Further, define
\[
\Gamma_{N+1}:= I \otimes \sigma_x,\quad \Gamma_{N+2}:= I\otimes \sigma_y.
\]
\end{enumerate}
It is easy to show that if the original $\Gamma_i$ satisfy~\eqref{acrel}, then so do the new ones. One can simply start the 
recursion with $\Gamma_1=1$ on the one-dimensional Hilbert space $\mathcal{H}_1:=\mathbb{C}$, and then apply the construction as 
often as necessary to obtain any finite number of matrices satisfying~\eqref{acrel}. For example, 
a single iteration gives the Pauli matrices
\begin{equation}
\label{pauli}
\Gamma_1=\sigma_z,\quad \Gamma_2=\sigma_x,\quad \Gamma_3 = \sigma_y,
\end{equation}
while after two iterations one has 
\begin{eqnarray*}
&\Gamma_1 = \sigma_z\otimes\sigma_z, \quad \Gamma_2 = \sigma_x\otimes\sigma_z,\\[5pt]
&\Gamma_3 = \sigma_y\otimes\sigma_z, \quad \Gamma_4 = I\otimes\sigma_x,\quad \Gamma_5 = I\otimes\sigma_y.
\end{eqnarray*}
The Clifford algebra relations~\eqref{acrel} have many interesting consequences. For example for $N\geq 2$, one has for any $k$ and $j\neq k$,
\begin{eqnarray*}
\Tr(\Gamma_k) &= \Tr(\Gamma_k\Gamma_j\Gamma_j) = -\Tr(\Gamma_j\Gamma_k\Gamma_j) \\[5pt]
&= -\Tr(\Gamma_k\Gamma_j\Gamma_j) = -\Tr(\Gamma_k),
\end{eqnarray*}
so that
\begin{equation}
\label{traceless}
\Tr(\Gamma_k)=0. 
\end{equation}
Another consequence is that
\begin{equation}
\label{cliffordprod}
\left(\sum_k X_k \Gamma_k\right)^2 = \left(\sum_k X_k^2\right) \cdot I
\end{equation}
for arbitrary real coefficients $X_k$.
\section*{Acknowledgments}
R.K. thanks the Perimeter Institute for hospitality during
his visit while this work was being performed. Research
at Perimeter Institute is supported by the Government of
Canada through Industry Canada and by the Province of
Ontario through the Ministry of Economic Development
and Innovation. C.H. was supported by the Engineering and
Physical Sciences Research Council. T.F. has been supported
by the John Templeton Foundation.

\end{document}